\documentclass[letterpaper, 10 pt, conference]{ieeeconf}

\IEEEoverridecommandlockouts 

\let\labelindent\relax

\usepackage{amsthm}

\usepackage{mathtools}
\usepackage{amsmath,bm}
\usepackage{moresize}
\usepackage{amssymb}
\usepackage{mathtools}
\usepackage{color}
\usepackage{verbatim}
\usepackage{mathrsfs}
\usepackage[normalem]{ulem}
\usepackage{enumitem}
\usepackage{hyperref}
\usepackage[ruled,vlined]{algorithm2e}
\usepackage{algpseudocode}
\usepackage{cite}

\definecolor{purple}{rgb}{0.5, 0.0, 0.5}
\definecolor{dark_green}{rgb}{0.0, 0.5, 0.0}
\definecolor{mygray}{gray}{0.6}
\definecolor{orange}{rgb}{1,0.5,0}

\DeclareMathAlphabet{\mathpzc}{OT1}{pzc}{m}{it}

\newcommand{\red}{\color{red}}
\newcommand{\green}{\color{green}}

\newcommand{\cyan}{\color{cyan}}

\newcommand{\white}{\color{white}}

\newcommand{\gray}{\color{mygray}}
\newcommand{\Cc}{\mathcal{C}}
\newcommand{\ow}{\mathcal{O}}

\newcommand{\Hh}{\mathcal{H}}
\newcommand{\I}{\mathcal{I}}

\newcommand{\J}{\mathcal{J}}

\newcommand{\Mb}{\bold{M}}

\newcommand{\Wb}{\bold{W}}

\newcommand{\yb}{\bold{y}}

\newcommand{\kbar}{\bar{k}}
\newcommand{\kt}{\tilde{k}}
\newcommand{\Gb}{\bold{G}}
\newcommand{\Gbb}{\bar{\Gb}}
\newcommand{\Gbt}{\tilde{\bold{G}}}

\newcommand{\Tb}{\bold{T}}

\newcommand{\T}{\mathcal{T}}
\newcommand{\Tbar}{\bar{T}}

\newcommand{\Z}{\mathbb{Z}}
\newcommand{\R}{\mathbb{R}}
\newcommand{\C}{\mathbb{C}}
\newcommand{\F}{\mathbb{F}}

\newcommand{\N}{\mathbb{N}}
\newcommand{\Q}{\mathbb{Q}}
\newcommand{\gb}{\bold{g}}

\newcommand{\pt}{\tilde{p}}

\newcommand{\Pb}{\bold{P}}

\newcommand{\Ab}{\bold{A}}
\newcommand{\Acal}{\mathcal{A}}
\newcommand{\Acalh}{\hat{\Acal}}
\newcommand{\Apzc}{\mathpzc{A}}

\newcommand{\Abt}{\tilde{\bold{A}}}

\newcommand{\bb}{\bold{b}}
\newcommand{\bbh}{\hat{\bold{b}}}

\newcommand{\Bb}{\bold{B}}
\newcommand{\Bbt}{\tilde{\bold{B}}}
\newcommand{\Bcal}{\mathcal{B}}
\newcommand{\Bcalb}{\bar{\mathcal{B}}}
\newcommand{\ab}{\bold{a}}

\newcommand{\Cb}{\bold{C}}
\newcommand{\cb}{\bold{c}}

\newcommand{\Dbt}{\tilde{\bold{D}}}

\newcommand{\Eb}{\bold{E}}

\newcommand{\eb}{\bold{e}}
\newcommand{\Ib}{\bold{I}}
\newcommand{\Hb}{\bold{H}}

\newcommand{\sfsty}[1]{\ensuremath{\mathsf{#1}}}  

\newcommand{\rsa}{\sfsty{RSA}}
\newcommand{\rs}{\sfsty{RS}}
\newcommand{\brs}{\sfsty{BRS}}
\newcommand{\x}{\sfsty{x}}

\DeclarePairedDelimiter\ceil{\lceil}{\rceil}
\DeclarePairedDelimiter\floor{\lfloor}{\rfloor}

\newtheorem{Thm}{Theorem}

\newtheorem{Prop}[Thm]{Proposition}
\newtheorem{Lemma}[Thm]{Lemma}

\newtheorem{Question}[Thm]{Question}

\DeclareMathOperator{\GL}{GL}

\newcommand{\qvec}[1]{\textbf{\textit{#1}}}

\begin{document}

\title{\LARGE \bf Secure Linear MDS Coded Matrix Inversion}

\author{$\textbf{Neophytos Charalambides}^{\mu}$, $\textbf{Mert Pilanci}^{\sigma}$, \textbf{and} $\textbf{Alfred O. Hero III}^{\mu}$\\
$\text{\white.}^{\mu}$EECS Department University of Michigan $\text{\white.}^{\sigma}$EE Department Stanford University\\
  Email: neochara@umich.edu, pilanci@stanford.edu, hero@umich.edu
\thanks{A preliminary version considered fractional repetition codes \cite{CPH20b}. All missing proofs can be found in \cite{CPH22a}, as well as an  encoding matrix illustration. This work was partially supported by grant ARO W911NF-15-1-0479.}
\vspace{-4mm}
}

\maketitle
\thispagestyle{empty}
\pagestyle{empty}


\begin{abstract}
A cumbersome operation in many scientific fields, is inverting large full-rank matrices. In this paper, we propose a coded computing approach for recovering matrix inverse approximations. We first present an approximate matrix inversion algorithm which does not require a matrix factorization, but uses a black-box least squares optimization solver as a subroutine, to give an estimate of the inverse of a real full-rank matrix. We then present a distributed framework for which our algorithm can be implemented, and show how we can leverage sparsest-balanced MDS generator matrices to devise \textit{matrix inversion coded computing schemes}. We focus on balanced Reed-Solomon codes, which are optimal in terms of computational load; and communication from the workers to the master server. We also discuss how our algorithms can be used to distributively compute the pseudoinverse of a full-rank matrix, and how the communication is secured from eavesdroppers.
\end{abstract}


\section{Introduction}
\label{intro}

Inverting a matrix is a common operation in numerous applications in domains such as social networks, numerical analysis and integration, machine learning, and scientific computing \cite{GS59,Hig02}. It is one of the most important operations, as it reverses a system. A common way of inverting a matrix is by performing Gaussian elimination, which in general takes $\ow(N^3)$ operations for square matrices of order $N$. In high-dimensional applications, this is cumbersome.

An operation of equivalent complexity, is multiplying two $N\times N$ matrices. The equivalency can be shown through the Schur complement. There is a plethora of efficient and elegant matrix multiplication algorithms, which imply matrix inversion algorithms. The most popular and practical algorithm; of complexity $\ow(N^{2.807})$, is due to Strassen \cite{Str69}. Many other inversion algorithms assume specific structure on the matrix, require a matrix-matrix product, or use a matrix factorization \cite{TB97}. Methods for matrix inversion or factorization are often referred to as \textit{direct methods}, in contrast to \textit{iterative methods}, which gradually converge to the solution \cite{DRSL16,PV20}. The most computationally efficient direct methods compute some form of the inverse, and are asymptotically equivalent. These have complexity $\ow(N^{\omega})$, for $\omega<2.373$ the \textit{matrix multiplication exponent} \cite{AV20}.

Distributed computations in the presence of \textit{stragglers} (workers who fail to compute their task or have longer response time than others) have gained a lot of attention in the information theory community. Coding-theoretic approaches have been adopted for this \cite{LLPPR17,YSRKSA18}, and fall under the framework of \textit{coded computing} (CC). Data security is also an increasingly important issue in CC \cite{LA20}. Despite the fact that multiplication algorithms imply inversion algorithms and vice versa, in the context of CC; matrix inversion has not been studied as extensively as \textit{coded matrix multiplication} (CMM) \cite{YMAA17}. The main reason for this is the fact that the latter is non-linear as an operator, which prohibits it from being parallelizable.  We point out that distributed inversion algorithms do exist, though these make assumptions on the matrix, are specific for distributed and parallel computing platforms, and require a matrix factorization; or heavy and multiple communication instances. In this work, we give a remedy to this, by approximating the columns of $\Ab^{-1}$, and using an encoding technique which has been leveraged in \textit{gradient coding} (GC) \cite{HASH17} to mitigate stragglers. We do not make any of the aforementioned assumptions.

Recall that the CC network is centralized, and is comprised of a master server who communicates with $n$ workers. The idea behind our approximation is that the workers use a least squares solver to approximate multiple columns of $\Ab^{-1}$. While other iterative procedures are applicable, we present simulation results with steepest descent (SD) and the conjugate gradient method (CG). By locally approximating the columns in this way, the workers can linearly encode the \textit{blocks} of $\widehat{\Ab^{-1}}$.

The non-linearity of matrix inversion prohibits linear or polynomial encoding of the data before the computations are to be carried out. Consequently, most CC approaches cannot be directly utilized. GC is the appropriate CC set up to consider \cite{TLDK17}, precisely because the encoding takes place once the computation which has been carried out, in contrast to most CMM schemes where the encoding is done by the master, before the data is distributed.

Once the workers complete their computations, they encode them by computing a linear combination with coefficients determined by a sparsest-balanced \textit{maximum distance separable} (MDS) generator matrix. This provides: 1) minimum redundancy per job across the network, 2) optimal communication from the workers to the master server. We focus on \textit{balanced Reed-Solomon} ($\brs$) code generator matrices \cite{HLH16,HLH16b}. Once a sufficient number of workers has responded, the master is able to recover the approximation $\widehat{\Ab^{-1}}$. We leverage the structure of sparsest-balanced generator matrices to optimally allocate tasks to the workers, while linear encoding results in minimal communication load from the workers to the master. The ideas discussed above are also extended to distributed approximation of the pseudoinverse $\Ab^{\dagger}$ for $\Ab$ full-rank, through a two or three-round communication CC approach.
We also present how the communication between the master and the workers can be made secure, guaranteeing security against eavesdroppers.

The paper is organized as follows. In Section \ref{prel_sec} we recall basic facts regarding matrix inversion, least squares approximation, and finite fields. In Section \ref{Appr_alg_sec} we present the matrix inverse and pseudoinverse approximation algorithms we utilize in our schemes. The main contribution is presented in Section \ref{Str_pr_sec}. We first review $\brs$ codes and then show how our inversion algorithm can be incorporated in linear CC schemes\footnote{For brevity, we abbreviate `coded computing scheme' to CCS.} derived from MDS sparsest-balanced generator matrices, with a focus on $\brs$ generator matrices. We then discuss how our pseudoinverse algorithm can be carried out distributively. Concluding remarks and future work are presented in Section \ref{concl_sec}.

\subsection{Related Work}

We point out two articles \cite{GHSTM11,YGK17} which have similarities to the matrix inverse approximation approach presented in this paper. Firstly, our approach to inverting $\Ab$ is similar in nature to \cite{GHSTM11}, which uses stochastic gradient descent to approximate matrix factorizations distributively. Secondly, the formulation of our underlying optimization problem: minimize $\|\Ab\Bb-\Ib_N\|_F^2$ by estimating the columns of $\Bb$, is equivalent to the problem studied in \cite{YGK17}, which deals with approximating linear inverse problems in the presence of stragglers. The drawbacks of the CCS provided in \cite{YGK17}, is that it is geared towards specific applications (e.g. personalized PageRank), makes assumptions on the covariance between the signals comprising the linear system and the accuracy of the workers, and assumes an additive decomposition of $\Ab$. Furthermore, the approximation algorithm in \cite{YGK17} is probabilistic, and considers the response of \textit{all} workers, treating stragglers as soft errors instead of erasures. We on the other hand make \textit{no} assumption on $\Ab$ other than the fact that it is non-singular, and our algorithm is not probabilistic.

The CC literature is vast, and has drawn ideas from many fields, e.g. graph theory, information theory, and optimization. We briefly discuss the most similar coding approaches to the one we propose, i.e. polynomial based codes.

Polynomial codes date back to 1960, with the invention of Reed-Solomon ($\rs$) codes \cite{RS60}. Variants of these codes have found application in many fields, and are still an active research area. In CC, polynomial codes have been used to devise CMM \cite{FC19,FJHDCG17,DFHJCG19,YA20,KD22}, as well as GC schemes \cite{HASH17,CPH20a}.

To multiply matrices $\Ab$ and $\Cb$, the ``MatDot'' CMM scheme \cite{FJHDCG17,DFHJCG19} uses an evaluation of a matrix polynomial as an encoding, whose coefficients are outer-products of the columns and rows of $\Ab$ and $\Cb$ respectively. Once a sufficient number of evaluations are sent back to the master server, she can apply a polynomial interpolation algorithm or $\rs$ decoding, in order to recover the coefficient which is equal to the product $\Ab\Cb$. The polynomial codes proposed in \cite{YMAA17,YA20} instead encode blocks of the rows and columns of $\Ab$ and $\Cb$ respectively. The workers then compute the product of the encodings they receive and send it back. Once sufficiently many jobs are received, an inversion of a Vandermonde matrix suffices for the decoding step.

The GC scheme from \cite{HASH17} is based on $\brs$ codes. The main difference to our work, is that in GC the objective is to construct an encoding matrix $\Gb$ and decoding vectors $\ab_\I\in\C^{k}$, such that $\ab_{\I}^\top\Gb=\vec{\bold{1}}$ for any set of non-straggling workers $\I$. The way $\brs$ codes are exploited in \cite{HASH17} is that we have the decomposition $\Gb_{\I}=\Hb_{\I}\Pb$, for $\Hb_{\I}$ a Vandermonde matrix, and the first row of $\Pb$ is equal to $\vec{\bold{1}}$. Therefore, $\ab_{\I}^\top$ is the first row of $\Hb_{\I}^{-1}$. The matrix subscripts $\I$, denote the submatrices of $\Gb$ and $\Hb$, consisting only of the rows indexed by $\I$. Further details on this CCS and how it differs from ours can be found in \cite{CPH22a} Appendix 2.

The state-of-the art CC framework is ``Lagrange Coded Computing'' (LCC), which is used to compute any arbitrary multivariate polynomial of a given dataset \cite{YSRKSA18,SMA21}. LCC is based on Lagrange interpolation, and it achieves the optimal trade-off between resiliency, security, and privacy. The problem we are considering is not a multivariate polynomial in terms of $\Ab$. To securely communicate $\Ab$ to the workers, we encode it through Lagrange interpolation. Though similar ideas appear in LCC, the purpose and application of the interpolation is different. Furthermore, LCC is a \textit{point-based} approach \cite{KD22} and requires additional interpolation and linear combination steps after the decoding takes place.

\section{Preliminary Background}
\label{prel_sec}

The set of $N\times N$ non-singular matrices is denoted by $\GL_N(\R)$. Recall that $\Ab\in\GL_N(\R)$ has a unique inverse $\Ab^{-1}$, such that $\Ab\Ab^{-1}=\Ib_N$. The simplest way of computing $\Ab^{-1}$ is by performing Gaussian elimination on $\big[\Ab|\Ib_N\big]$, which gives $\big[\Ib_N\big|\Ab^{-1}]$ in $\ow(N^3)$ operations. In Algorithm \ref{inv_alg}, we approximate $\Ab^{-1}$ column-by-column. We denote the $i^{th}$ row and column of $\Ab$ respectively by $\Ab_{(i)}$ and $\Ab^{(i)}$.

For full-rank rectangular matrices $\Ab\in\R^{N\times M}$ where $N>M$, one resorts to the left Moore–Penrose pseudoinverse $\Ab^{\dagger}\in\R^{M\times N}$, for which $\Ab^{\dagger}\Ab=\Ib_M$. In Algorithm \ref{pinv_alg}, we present how to approximate the left pseudoinverse of $\Ab$, by using the fact that $\Ab^{\dagger}=(\Ab^\top\Ab)^{-1}\Ab^\top$; since $\Ab^\top\Ab\in\GL_N(\R)$. The right pseudoinverse $\Ab^{\dagger}=\Ab^\top(\Ab\Ab^\top)^{-1}$ of $\Ab\in\R^{M\times N}$ where $M<N$, can be obtained by a modification of Algorithm \ref{pinv_alg}.

In the proposed algorithms we approximate $N$ instances of the least squares minimization problem
\begin{equation}
\label{OLS}  
  \theta^{\star}_{ls} = \arg\min_{\theta\in\R^M} \left\{\|\Ab\theta-\yb\|_2^2\right\} 
\end{equation}
for $\Ab\in\R^{N\times M}$ and $\yb\in\R^N$. In many applications $N\gg M$, where the rows represent the feature vectors of a dataset. This has the closed-form solution $\theta^{\star}_{ls} = \Ab^{\dagger}\yb$.

Computing $\Ab^{\dagger}$ to solve \eqref{OLS} directly is intractable for large $M$, as it requires computing the inverse of $\Ab^\top\Ab$. Instead, we use gradient methods to get \textit{approximate} solutions, e.g. SD or CG, which require less operations, and can be done distributively. One could use second-order methods; e.g. Newton–Raphson, Gauss-Newton, Quasi-Newton, BFGS, or Krylov subspace methods instead. This is not the focus of our work. We denote the iteration count of these methods with a superscript $[t]$, for $t=1,2,3,...$ .

Our schemes are defined over the finite field of $q$ elements, $\F_q$. We denote its cyclic multiplicative subgroup by $\F_q^{\times}=\F_q\backslash\{0_{\F_q}\}$. For implementation purposes, we identify finite fields with their realization in $\C$ as a subgroup of the circle group, since we assume our data is over $\R$. All operations can therefore be carried out over $\C$. Specifically, for $\beta\in\F_q^{\times}$ a generator, we identify $\beta^j$ with $e^{2\pi ij/q}$, and $0_{\F_q}$ with $1$. The set of integers between $1$ and $\nu$ is denoted by $\N_\nu$.

\subsection{Balanced Reed-Solomon Codes}
\label{BRS_intro}

A Reed-Solomon code $\rs_q[n,k]$ over $\F_q$ for $q>n>k$, is the encoding of polynomials of degree at most $k-1$, for $k$ the message length and $n$ the code length. It represents our message over the \textit{defining set of points} $\Apzc=\{\alpha_i\}_{i=1}^n\subset\F_q$
\begin{align*}
  \rs_q[n,k]=\Big\{\big[f&(\alpha_1),f(\alpha_2),\cdots,f(\alpha_n)\big] \ \Big|\\
  & f(X)\in\F_q[X] \text{ of degree }\leqslant k-1 \Big\}
\end{align*}
where $\alpha_i=\alpha^i$, for $\alpha$ a primitive root of $\F_q$. Hence, each $\alpha_i$ is distinct. A natural interpretation of $\rs_q[n,k]$ is through its encoding map. Each message $(m_0,...,m_{k-1})\in\F_q^k$ is interpreted as $f(\x)=\sum_{i=0}^{k-1}m_i\x^i\in\F_q[\x]$, and $f$ is evaluated at each point of $\Apzc$. From this, $\rs_q[n,k]$ can be defined through a generator matrix
$$ \Gb = \begin{pmatrix} 1 & \alpha_1 & \alpha_1^2 & \hdots & \alpha_1^{k-1} \\ 1 & \alpha_2 & \alpha_2^2 & \hdots & \alpha_2^{k-1} \\ \vdots & \vdots & \vdots & \ddots & \vdots \\ 1 & \alpha_n & \alpha_n^2 & \hdots & \alpha_n^{k-1} \end{pmatrix} \in \F_q^{n\times k} , $$
thus, $\rs$ codes are linear codes over $\F_q$. Furthermore, they obtain the Singleton bound, i.e. $d=n-k+1$ where $d$ is the code's distance, which means they are MDS.

Balanced Reed-Solomon codes \cite{HLH16,HLH16b} are a family of linear MDS error-correcting codes with generator matrices $\Gb$ that are:
\begin{itemize}
  \item \textbf{sparsest}: each \textit{row} has the least possible number of nonzero entries
  \item \textbf{balanced}: each \textit{column} contains the same number of nonzero entries
\end{itemize}
for the given code parameters $k$ and $n$. The design of these generators are suitable for our purposes, as:
\begin{enumerate}
  \item we have balanced loads across homogeneous workers,
  \item sparse generator matrices means we have reduced computational tasks across the network,
  \item the MDS property permits an efficient decoding step,
  \item linear codes produce a compressed representation of the encoded blocks.
\end{enumerate}

\section{Approximation Algorithms}
\label{Appr_alg_sec}

\subsection{Proposed Inverse Algorithm}
\label{pr_inv_alg_subsec}

Our goal is to estimate $\Ab^{-1}=\big[\bb_1 \ \cdots \ \bb_N \big]$, for $\Ab$ a square matrix of order $N$. A key property to note is
$$ \Ab\Ab^{-1} = \Ab\big[\bb_1 \ \cdots \ \bb_N \big] = \big[\Ab\bb_1 \ \cdots \ \Ab\bb_N \big] = \bold{I}_N $$
which implies that $\Ab\bb_i=\eb_i$ for all $i\in\N_N$, where $\eb_i$ are the standard basis column vectors. Assume for now that we use any black-box least squares solver to estimate
\begin{equation}
\label{inv_LS}
  \hat{\bb}_i \approx \arg\min_{\bb\in\R^N} \Big{\{f_i(\bb)\coloneqq\|\Ab\bb-\eb_i\|_2^2}\Big\}
\end{equation}
which we call $N$ times, to recover $\widehat{\Ab^{-1}} \coloneqq \big[ \hat{\bb}_1 \ \cdots \ \hat{\bb}_N \big]$. This approach may be viewed as solving
\begin{equation*}
\label{inv_LS_F}
  \widehat{\Ab^{-1}} \approx \arg\min_{\ \ \Bb\in\R^{N\times N}}\left\{\|\Ab\Bb-\bold{I}_N\|_F^2\right\}.
\end{equation*}
Alternatively, one could estimate the rows of $\Ab^{-1}$. Algorithm \ref{inv_alg} shows how this can be performed by a single server.

\begin{algorithm}[h]
\label{inv_alg}
\SetAlgoLined
\KwIn{$\Ab\in\GL_N(\R)$}
  \For{i=1 to N}
  {
    approximate $\hat{\bb}_i \approx \arg\min_{\bb\in\R^N} \left\{\|\Ab\bb-\eb_i\|_2^2\right\}$ 
  }
 \Return $\widehat{\Ab^{-1}} \gets \big[ \hat{\bb}_1 \ \cdots \ \hat{\bb}_N \big]$
 \caption{Estimating $\Ab^{-1}$}
\end{algorithm}

In the case where SD is used to approximate $\hat{\bb}_i$ from \eqref{inv_LS}, the overall operation count is $\ow(\T_i N^2)$; for $\T_i$ the total number of descent iterations used. An upper bound on the number of iterations can be determined by the underlying termination criterion, e.g. the criterion $f_i(\bbh^{[t]})-f_i(\bb^{\star})\leqslant\epsilon$ is guaranteed to be satisfied after $\T=\ow(\log(1/\epsilon))$ iterations \cite{BV04}. The overall error of $\widehat{\Ab^{-1}}$ may be quantified as
\begin{itemize}
  \item $\text{err}_{\ell_2}(\widehat{\Ab^{-1}}) \coloneqq \|\widehat{\Ab^{-1}}-\Ab^{-1}\|_2$
  \item $\text{err}_F(\widehat{\Ab^{-1}}) \coloneqq \|\widehat{\Ab^{-1}}-\Ab^{-1}\|_F$
  \item $\text{err}_{\text{r}F}(\widehat{\Ab^{-1}}) \coloneqq \frac{\|\widehat{\Ab^{-1}}-\Ab^{-1}\|_F}{\|\Ab^{-1}\|_F} = \frac{\sum\limits_{i=1}^{N} \|\Ab\hat{\bb}_i-\eb_i\|_2}{\|\Ab^{-1}\|_F}$
\end{itemize}
which we refer to as the \textit{$\ell_2$-error}, \textit{Frobenius-error} and \textit{relative Frobenius-error} respectively. The corresponding pseudoinverse approximation errors are defined accordingly.

To compute $\widehat{\Ab^{-1}}$ distributively, each of the $n$ servers are asked to estimate $T$-many $\bbh_i$'s in parallel. When using SD, the worst-case runtime by the workers is $\ow(T\cdot\T_{\text{max}}N^2)$, for $\T_{\text{max}}$ the maximum number of iterations of SD among the workers. If CG is used, each worker needs no more than $NT$ CG steps to exactly compute its task; i.e. $\ow\big(TN\frac{\sigma_{\max}(\Ab)}{\sigma_{\min}(\Ab)}\big)$ operations, which is the worst case runtime \cite{She94,TB97}.

Bounds on $\text{err}_F(\widehat{\Ab^{-1}})$ and $\text{err}_{\text{r}F}(\widehat{\Ab^{-1}})$ can be established for both algorithms, specific to the black-box least squares solver being utilized. This is left for future work.

\subsection{Proposed Pseudoinverse Algorithm}
\label{pr_pinv_alg_subsec}

Similar to the inverse, the pseudoinverse of a matrix also appears in a variety of applications. Computing the pseudoinverse of $\Ab\in\R^{N\times M}$ for $N>M$ is even more cumbersome, as it requires inverting the Gram matrix $\Ab^\top\Ab$. For this subsection, we consider a full-rank rectangular matrix $\Ab$.

One could naively attempt to modify Algorithm \ref{inv_alg} in order to retrieve $\widehat{\Ab^{\dagger}}$ such that $\widehat{\Ab^{\dagger}}\Ab\approx\Ib_M$, by approximating the rows of $\Ab^{\dagger}$. This would \textit{not} work, as the underlying optimization problems would not be strictly convex. Instead, we use Algorithm \ref{pinv_alg} to estimate the rows of $\Bb^{-1}\coloneqq(\Ab^\top\Ab)^{-1}$, and then multiply the estimate $\widehat{\Bb^{-1}}$ by $\Ab^\top$. This gives us the approximation $\widehat{\Ab^{\dagger}}\coloneqq\widehat{\Bb^{-1}}\cdot\Ab^\top$.

The drawback of Algorithm \ref{pinv_alg} is that it requires two additional matrix multiplications, $\Ab^\top\Ab$ and $\widehat{\Bb^{-1}}\Ab^\top$. We overcome this barrier by using a CMM scheme twice, to recover $\widehat{\Ab^{\dagger}}$ in a two or three-round communication CC approach. These are discussed in \ref{distr_Pseudinv_subsec}.

\begin{algorithm}[h]
\label{pinv_alg}
\SetAlgoLined
\KwIn{full-rank $\Ab\in\R^{N\times M}$ where $N>M$}
  $\Bb\gets \Ab^\top\Ab$\\
  \For{i=1 to M}
  {
    $\hat{\cb}_i \approx \arg\min_{\cb\in\R^{1\times M}} \Big\{g_i(\cb)\coloneqq\|\cb\Bb-\eb_i^\top\|_2^2\Big\}$\\ 
    $\hat{\bb}_i\gets \hat{\cb}_i\cdot\Ab^\top$
  }
 \Return $\widehat{\Ab^{\dagger}} \gets \left[ \hat{\bb}_1^\top \ \cdots \ \hat{\bb}_M^\top \right]^\top$ \Comment{$\widehat{\Ab^{\dagger}}_{(i)}=\bbh_i$}
 \caption{Estimating $\Ab^{\dagger}$}
\end{algorithm}

\vspace{-4mm}
\subsection{Numerical Experiments}
\label{exper_sec}

The accuracy of the proposed algorithms was tested on randomly generated matrices, using both SD and CG \cite{TB97} for the subroutine optimization problems. The depicted results are averages of 20 runs, with termination criteria $\|\nabla f_i(\bb^{[t]})\|_2\leqslant \epsilon$ for SD and $\|\bb_i^{[t]}-\bb_i^{[t-1]}\|_2\leqslant \epsilon$ for CG, for the given $\epsilon$ accuracy parameters. The criteria for $\widehat{\Ab^{\dagger}}$ were analogous. We considered $\Ab\in \R^{100\times 100}$ and $\Ab\in\R^{100\times 50}$. The error subscripts represent $\mathscr{A}=\{\ell_2,{F},\text{r}F\}$, $\mathscr{N}=\{\ell_2,F\}$, $\mathscr{F}=\{F,\text{r}F\}$. We note that significantly fewer iterations took place when CG was used for the same $\epsilon$, though this depends heavily on the choice of the step-size. Thus, there is a trade-off between accuracy and speed when using SD vs. CG, for such termination criteria. 

\begin{center}
\begin{tabular}{ |p{.5cm}||p{1.05cm}|p{1.1cm}|p{1.1cm}|p{1.1cm}|p{1.1cm}| }
\hline
\multicolumn{6}{|c|}{Average $\widehat{\Ab^{-1}}$ errors, for $\Ab\sim50\cdot \mathcal{N}(0,1)$ --- SD} \\
\hline
$\epsilon$ & $10^{-1}$ & $10^{-2}$ & $10^{-3}$ & $10^{-4}$ & $10^{-5}$ \\
\hline
$\text{err}_{\mathscr{A}}$ & {\small$\ow(10^{-2})$} & {\small$\ow(10^{-5})$} & {\small$\ow(10^{-7})$} & {\small$\ow(10^{-9})$} & {\small$\ow(10^{-12})$} \\
\hline
\end{tabular}
\end{center}

\begin{center}
\begin{tabular}{ |p{.5cm}||p{1.05cm}|p{1.05cm}|p{1.05cm}|p{1.1cm}|p{1.1cm}| }
\hline
\multicolumn{6}{|c|}{{\small Average $\widehat{\Ab^{-1}}$ errors, for $\Ab\sim50\cdot \mathcal{N}(0,1)$ --- CG}} \\
\hline
$\epsilon$ & $10^{-3}$ & $10^{-4}$ & $10^{-5}$ & $10^{-6}$ & $10^{-7}$ \\
\hline
$\text{err}_{\mathscr{N}}$ & {\small$\ow(10^{-3})$} & {\small$\ow(10^{-5})$} & {\small$\ow(10^{-8})$} & {\small$\ow(10^{-11})$} & {\small$\ow(10^{-12})$} \\
$\text{err}_{\text{r}F}$ & {\small$\ow(10^{-3})$} & {\small$\ow(10^{-5})$} & {\small$\ow(10^{-7})$} & {\small$\ow(10^{-10})$} & {\small$\ow(10^{-12})$} \\
\hline
\end{tabular}
\end{center}

\begin{center}
\begin{tabular}{ |p{.5cm}||p{1.05cm}|p{1.05cm}|p{1.05cm}|p{1.15cm}|p{1.15cm}| }
\hline
\multicolumn{6}{|c|}{Average $\widehat{\Ab^{\dagger}}$ errors, for $\Ab\sim \mathcal{N}(0,1)$ --- SD} \\
\hline
$\epsilon$ & $10^{-1}$ & $10^{-2}$ & $10^{-3}$ & $10^{-4}$ & $10^{-5}$ \\
\hline
$\text{err}_{\ell_2}$ & {\small$\ow(10^{-4})$} & {\small$\ow(10^{-6})$} & {\small$\ow(10^{-8})$} & {\small$\ow(10^{-10})$} & {\small$\ow(10^{-12})$} \\
$\text{err}_{\mathscr{F}}$ & {\small$\ow(10^{-5})$} & {\small$\ow(10^{-7})$} & {\small$\ow(10^{-9})$} & {\small$\ow(10^{-11})$} & {\small$\ow(10^{-13})$} \\
\hline
\end{tabular}
\end{center}

\begin{center}
\begin{tabular}{ |p{.5cm}||p{1.05cm}|p{1.05cm}|p{1.05cm}|p{1.15cm}|p{1.15cm}| }
\hline
\multicolumn{6}{|c|}{Average $\widehat{\Ab^{\dagger}}$ errors, for $\Ab\sim \mathcal{N}(0,1)$ --- CG} \\
\hline
$\epsilon$ & $10^{-3}$ & $10^{-4}$ & $10^{-5}$ & $10^{-6}$ & $10^{-7}$ \\
\hline
$\text{err}_{\ell_2}$ & {\small$\ow(10^{-4})$} & {\small$\ow(10^{-6})$} & {\small$\ow(10^{-8})$} & {\small$\ow(10^{-10})$} & {\small$\ow(10^{-12})$} \\
$\text{err}_{\mathscr{F}}$ & {\small$\ow(10^{-2})$} & {\small$\ow(10^{-3})$} & {\small$\ow(10^{-8})$} & {\small$\ow(10^{-10})$} & {\small$\ow(10^{-12})$} \\
\hline
\end{tabular}
\end{center}

\section{Coded Matrix Inversion}
\label{Str_pr_sec}  

In this section, we focus on CC and give a linear scheme based on $\brs$ codes \cite{HLH16,HLH16b,HASH17} which makes Algorithms \ref{inv_alg} and \ref{pinv_alg} resilient to stragglers. We present the proposed scheme for Algorithm \ref{inv_alg}, and then show how to combine Polynomial CMM \cite{YMAA17}; to distributively perform Algorithm \ref{pinv_alg}. While there is extensive literature on matrix-matrix, matrix-vector multiplication, and computing the gradient in the presence of stragglers, there is limited work on computing or approximating the inverse of a matrix \cite{YGK17}. 

First, we argue why all of $\Ab$ needs to be known by each of the workers, in order to recover entries or columns of its inverse. We then show how Lagrange interpolation can be utilized to securely share $\Ab$ among the workers. We then discuss what are the computational tasks the workers are requested to compute, which are blocks of $\widehat{\Ab^{-1}}$; and correspond to the subroutine problems of Algorithms \ref{inv_alg}, \ref{pinv_alg}.

Then, we briefly review $\brs$ codes, how the workers encode their computations in our proposed CCS, and how the master then decodes the received computations. Optimality of $\brs$ generator matrices in terms of allocated tasks and encoded communication loads are also established, in Lemma \ref{lem_sol_opt_prob} and in \ref{opt_BRS_subs} respectively.

We note that when assuming finite-point arithmetic, the CCS we propose introduces no numerical nor approximation errors. The approximation in our procedure, is a consequence of using iterative solvers to estimate \eqref{inv_LS}. Therefore, if the workers can recover the optimal solutions to the underlying minimization problems, our approach would be \textit{exact}.

\subsection{Encrypting and Communicating $\Ab$}
\label{enc_A_subs}

A bottleneck when computing the inverse of a matrix; or estimating its columns, is that the entire matrix needs to be known. A single change in the matrix's entries may result in a non-singular matrix. Below, we illustrate a simple such example. If we change $\Ab_{2,3}$ of $\Ab$ for which $\text{rank}(\Ab)=3$:
\begin{equation}
\label{perturbed_example}
  \Ab = \begin{pmatrix} 8 & 2 & 5\\ 2 & 2 & \textit{\cyan 5}\\ 3 & 7 & 5 \end{pmatrix} \quad \leadsto \quad \qvec{A}^{\bullet} = \begin{pmatrix} 8 & 2 & 5\\ 2 & 2 & \textit{\cyan 2}\\ 3 & 7 & 5 \end{pmatrix}
\end{equation}
we get $\qvec{A}^{\bullet}$ for which $\text{rank}(\qvec{A}^{\bullet})=2$. This conveys how sensitive Gaussian elimination is \cite{ColNotes}.

In the case where only one column is not known, one can determine the subspace in which the missing column lies in, but without the knowledge of at least one entry of that column, it would be impossible to recover that column. Even with such an approach or a matrix completion algorithm, the entire $\Ab$ is determined before we proceed to inverting $\Ab$, or performing linear regression to solve $\Ab\bb=\eb_i$. Problems similar to the one illustrated in \eqref{perturbed_example} are extensively studied in conditioning and stability of numerical analysis \cite{TB97}, and in perturbation theory. This is not a focus of our work.

Furthermore, by the data processing inequality \cite[Corollary pg.35]{CT06}, the above imply that no less than $N^2$ information symbols can be delivered to each worker, while hoping to approximate a column of $\Ab^{-1}$, if no assumption is to be made on the structure of $\Ab$. Hence, we cannot deliver a representation of $\Ab$ with less than $N^2$ symbols. This is a consequence of the fact that a dense vector is not recoverable from underdetermined linear measurements. We can however send an encoded version of $\Ab$ to the workers consisting of $N^2$ symbols, determined by a modified Lagrange polynomial, which guarantees security against eavesdroppers.

Similar cryptographic protocols date back to Shamir's secret sharing scheme \cite{Sha79}, which is also based on $\rs$ codes. More recently, this idea has extensively been exploited in LCC \cite{YSRKSA18}. The way it is used in LCC differs from ours, as we need knowledge of the entire matrix $\Ab$.

Let $k$ be a positive factor of $N$ and $T=\frac{N}{k}$.\footnote{If $k\nmid N$, append $\bold{0}_{N\times 1}$ to the end of the first $\kt=\text{rem}(N,k)$ blocks which are each comprised of $\tilde{T}=\floor{\frac{N}{k}}$ columns of $\Ab$, while the remaining $k-\kt$ blocks are comprised of $\tilde{T}+1$ columns. Now, each block is of size $N\times(\tilde{T}+1)$.} Select a set of distinct \textit{interpolation points} $\Bcal=\{\beta_j\}_{j=1}^n\subsetneq \F_q^{\times}$, for $q>n$.\footnote{For the encoding of $\Ab$, $k$ points suffice, and we only need to require $q>k$. We select $\Bcal$ of cardinality $n$ and require $q>n$, in order to also use $\Bcal$ in our CCS.} To construct this set, sample $\beta\in\F_q^{\times}$; any one of the $\phi(q-1)$ primitive roots of $\F_q$ ($\phi$ is Euler's totient function), which is a generator of the multiplicative group $(\F_q^{\times},\cdot)$, and define each point as $\beta_j=\beta^j$. We then generate a random multiset $\Hh=\{\eta_j\}_{j=1}^k\in2^{\F_q^{\times}}$ of size $k$, i.e. repetitions in $\Hh$ are allowed, which we will use to remove the structure of the Lagrange coefficients, as the adversaries could use them to reveal $\beta$.

The element $\beta$ and set $\Hh^{-1}\coloneqq\{\eta_j^{-1}\}_{j=1}^k$, are broadcasted securely to all the workers through a public-key cryptosystem, e.g. $\rsa$. Matrix $\Ab$ is then partitioned into $k$ blocks
\begin{equation}
\label{parts_A}
  \Ab=\Big[\Ab_1 \ \cdots \ \Ab_k\Big] \ \ \text{ where } \Ab_i\in\R^{N\times T},\ \forall i\in\N_k.
\end{equation}
Next, $\Ab$ is encoded through the univariate polynomial
\begin{equation}
\label{lagr_pol_matr}
  f(\x) = \sum\limits_{j=1}^k\Ab_j\cdot\eta_j\left(\prod\limits_{l\neq j}\frac{\x-\beta_l}{\beta_j-\beta_l}\right)
\end{equation}
for which $f(\beta_j)=\eta_j\Ab_j$. This is then shared with the workers, who recover $\Ab$ as follows:
$$ \Ab=\Big[\eta_1^{-1}f(\beta_1) \ \cdots \ \eta_k^{-1}f(\beta_k)\Big]\in\R^{N\times N}. $$
The coefficients of $f(\x)$ are comprised of $NT$ symbols, thus, the polynomial consists of a total of $N^2$ symbols.

\begin{Prop}
\label{prop_sec_cryptosystem}
The encryption of $\Ab$ through $f(\x)$, is as secure against eavesdroppers as the public-key cryptosystem which was used to broadcast $\beta$ and $\Hh^{-1}$.
\end{Prop}

For an additional security layer, the interpolation points of $\Bcal$ could instead be defined as $\beta_{j}=\beta^{\pi(j)}$, for $\pi\in S_n$ a random permutation. In this case, $\pi^{-1}$ also needs to be securely broadcasted, so that the workers can determine $\Bcal$.

\subsection{Computational Tasks}

For Algorithms \ref{inv_alg} and \ref{pinv_alg}, any CCS in which the workers compute an encoding of partitions of the resulting computation $\Eb=\big[E_1 \ \cdots \ E_k \big]$ could be utilized. It is crucial that the encoding takes place on the computed tasks $\{E_i\}_{i=1}^k$ in the scheme, and \textit{not} the assigned data or partitions of the matrices that are being computed over (e.g. \cite{YA20}), otherwise the algorithms could potentially not return the correct results. This also means that utilizing such encryption approaches (e.g. \cite{YSRKSA18}) for guaranteeing security against the workers, is not an option. Such schemes leverage the linearity of matrix multiplication. We face these restrictions due to the fact that matrix inversion is a non-linear operator.

The computation tasks $E_i$ correspond to a partitioning $ \widehat{\Ab^{-1}} = \big[\Acalh_1 \ \cdots \ \Acalh_k \big]$, of our approximation from Algorithm \ref{inv_alg}. We propose a linear encoding of the computed blocks $\{\Acalh_i\}_{i=1}^k$ in \ref{BRS_CC_subs}. Along with the proposed decoding step, we have a MDS CCS for matrix inversion.

We consider the same parameters as in \ref{enc_A_subs}, in order to reuse $\Bcal$ in our CCS. Each $\Acalh_i$ is comprised of $T$ distinct but consecutive approximations of \eqref{inv_LS}, i.e.
$$ \Acalh_i = \big[\bbh_{(i-1)T+1}\ \cdots \ \bbh_{iT}\big]\in\R^{N\times T} \quad \forall i\in\N_k, $$
which could also be approximated by iteratively solving
$$ \Acalh_i \approx \arg\hspace{-3mm}\min_{\ \ \Bb\in\R^{N\times T}} \hspace{-1mm} \left\{\left\|\Ab\Bb-\big[\eb_{(i-1)T+1}\ \cdots \ \eb_{iT}\big]\right\|_F^2\right\}. $$

We assume the workers are \textit{homogeneous}, i.e. they have the same computational power. Therefore, equal computational loads are assigned to each of them. Without loss of generality, we assume that the workers use the same algorithms and parameters for estimating the columns $\{\bbh_i\}_{i=1}^N$. Therefore, workers allocated the same tasks are expected to get equal approximations in the same amount of time.

\subsection{Balanced Reed-Solomon Codes for CC}
\label{BRS_CC_subs}

Recall that we leverage $\brs$ generator matrices for our CC inversion scheme. For simplicity, we will consider the case where $d=s+1=\frac{nw}{k}$ is a positive integer\footnote{The case where $\frac{nw}{k}\in\Q_+\backslash\Z_+$ is analysed in \cite{HASH17}, and also applies to our approach. We restrict our discussion to the case where $\frac{nw}{k}\in\Z_+$.}, for $n$ the number of workers and $s$ the number of stragglers. Furthermore, $d$ is the distance of the code and $\|\Gb^{(j)}\|_0=d$ for all $j\in\N_k$; $\|\Gb_{(i)}\|_0=w$ for all $i\in\N_n$, and $d>w$ since $n>k$. For decoding purposes, we require that at least $k=n-s$ workers respond. Consequently, $d=s+1$ implies that $n-d=k-1$. For simplicity, we also assume $d\geqslant n/2$.

For conventional reasons we use the transpose of $\brs$ generator matrices, so from here on we consider such generator matrices $\Gb\in\F_q^{n\times k}$. In our setting, each column of $\Gb$ corresponds to a computational task of $\widehat{\Ab^{-1}}$; i.e. a block $\Acalh_i$, and each row corresponds to a worker.

Our choice of such a generator matrix $\Gb\in\F_q^{n\times k}$ solves the minimization problem
\begin{equation}
\label{min_G_problem}
\begin{aligned}
\arg\min_{\Gb\in\F_q^{n\times k}} \quad & \big\{\text{nnzr}(\Gb)\big\}\\
\textrm{s.t.} \quad & \|\Gb_{(i)}\|_0=w,\ \forall i\in\N_n\\
  & \|\Gb^{(j)}\|_0=d,\ \forall i\in\N_k\\
  & \text{rank}(\Gb_\I)=k,\ \forall \I:|\I|=k\\
\end{aligned}
\end{equation}
which determines an optimal task allocation among the workers of our CCS.

Under the above assumptions, the entries of the generator matrix of a $\brs_q[n,k]$ code meet the following:
\begin{itemize}
  \item each \textit{column} is sparsest, with exactly $d$ nonzero entries
  \item each \textit{row} is balanced, with $w=\frac{dk}{n}$ nonzero entries
\end{itemize}
where $d$ equals to the number of workers who are tasked to compute each block, and $w$ is the number of blocks which are computed by each worker.

Each column $\Gb^{(j)}$ corresponds to a polynomial $p_j(\x)$, whose entries are the evaluation of the polynomial at each of the points of the defining set $\Apzc$ defined in \ref{BRS_intro}, i.e. $\Gb_{ij}=p_j(\alpha_i)$ for $(i,j)\in\N_n\times\N_k$. To construct the polynomials $\{p_j(\x)\}_{i=1}^k$, for which $\text{deg}(p_j)\leqslant\text{nnzr}(\Gb^{(j)})=n-d=k-1$, we first need to determine a sparsest and balanced \textit{mask matrix} $\Mb\in\{0,1\}^{n\times k}$, which is $\rho$-sparse for $\rho=\frac{d}{n}$; i.e. $\text{nnzr}(\Gb)=\rho nk$. It is fairly easy to construct such matrices, by using the Gale-Ryser Theorem \cite{DSDY13,Kra96}. Furthermore, deterministic constructions resemble generator matrices of cyclic codes.

For our purposes we use $\Bcal$ as our defining set of points, where each point corresponds to the worker with the same index. The objective now is to devise the polynomials $p_j(\x)$, for which $p_j(\beta_i)=0$ if and only if $\Mb_{ij}=0$. Therefore:
\begin{enumerate}[label=(\roman*)]
  \item $\Mb_{ij}=0 \quad \implies \quad (\x-\beta_i)\mid p_j(\x)$
  \item $\Mb_{ij}\neq0 \quad \implies \quad p_j(\beta_i)\in\F_q^{\times}$
\end{enumerate}
for all pairs $(i,j)$.

The construction of $\brs[n,k]_q$ from \cite{HLH16} is based on what the authors called \textit{scaled polynomials}. Below, we summarize the construction given in \cite{HASH17}, which is based on Lagrange interpolation. We then prove a simple but important fact about it, which allows us to perform our decoding step.

The univariate polynomials corresponding to each column $\Gb^{(j)}$, are defined as:
\begin{equation}
\label{lagr_polys}
  p_j(\x) \coloneqq \prod\limits_{i:\Mb_{ij}=0}\left(\frac{\x-\beta_i}{\beta_j-\beta_i}\right) = \sum\limits_{\iota=1}^{k}p_{j,\iota}\cdot\x^{\iota-1}\in\F_q[\x]
\end{equation}
which satisfy (i) and (ii). By the sparsity parameters of $\Mb$ and the BCH bound \cite[Chapter 9]{Mc01}, it follows that $\text{deg}(p_j)\geqslant n-d=k-1$ for all $j\in\N_k$. Since each $p_j(\x)$ is the product of $n-d$ monomials, we conclude that the bound on the degree is satisfied and met with equality, hence $p_{j,\iota}\in\F_q^{\times}$ for all coefficients. 

By the construction of $\Gb$, both $\Gb$ and $\Gb_{\I}$ are decomposable into a Vandermonde matrix $\Hb\in\Bcal^{n\times k}$ and a matrix comprised of the polynomial coefficients $\Hb\in(\F_q^{\times})^{k\times k}$ \cite{HASH17}. Specifically, $\Gb=\Hb\Pb$ where $\Hb_{ij}=\beta_i^{j-1}=\beta^{i(j-1)}$ and $\Pb_{ij}=p_{j,i}$ are the coefficients from \eqref{lagr_polys}. This can be interpreted as $\Pb^{(j)}$ defining the polynomial $p_j(\x)$, and $\Hb_{(i)}$ is comprised of the first $k$ positive powers of $\beta_i$ in ascending order, therefore
\vspace{-2mm}
$$ p_j(\beta_i) = \sum\limits_{\iota=1}^{k}p_{j,\iota}\cdot \beta_i^{\iota-1} = \langle\Hb_{(i)},\Pb^{(j)}\rangle . $$

\begin{Lemma}
\label{inverse_lem}
The restriction $\Gb_{\I}\in\F_q^{k\times k}$ of $\Gb$ to any of its $k$ rows indexed by $\I\subsetneq\N_n$, is invertible. Moreover, its inverse can be computed online in $\ow(k^2+k^{\omega})$ operations.\footnote{Recall that $\omega<2.373$ is the matrix multiplication exponent.}
\end{Lemma}

\begin{Lemma}
\label{lem_sol_opt_prob}
The generator matrix $\Gb\in\F_q^{n\times k}$ of a $\brs_q[n,k]$ MDS code defined by the polynomials $p_j(\x)$ of \eqref{lagr_polys}, solves the optimization problem \eqref{min_G_problem}.
\end{Lemma}

Lemma \ref{inverse_lem} implies that as long as $k$ workers respond, the approximation $\widehat{\Ab^{-1}}$ is recoverable. Moreover, the decoding step reduces to a matrix multiplication of $k\times k$ matrices. Applying $\Hb_\I$ to a square matrix can be done in $\ow(k^2\log k)$ through the FFT algorithm. The prevailing computation in our decoding, is applying $\Pb^{-1}$.

\subsection{Coded Matrix Inversion Scheme}

For our CCS, we utilize $\brs$ generator matrices for both the encoding and decoding steps. We adapt the GC framework, so we need an analogous condition to $\ab_{\I}^\top\Gb=\vec{\bold{1}}$ for \textit{coded matrix inversion}; in order to invoke Algorithm \ref{inv_alg}. The condition we require is $\Dbt_\I\Gbt=\Ib_N$, for an encoding-decoding pair $(\Gbt,\Dbt_\I)$.

From our discussion on $\brs$ codes, we set $\Gbt=\Ib_T\otimes\Gb$ and $\Dbt_\I=\Ib_T\otimes(\Gb_\I)^{-1}$ for any given set of $k$ responsive workers indexed by $\I$. The index set of blocks requested from the $\iota^{th}$ worker to compute is denoted by $\J_\iota$, and has cardinality $w$. The encoding steps correspond to
\begin{equation}
\label{enc_identity}
  \Gbt\cdot(\widehat{\Ab^{-1}})^\top = (\Ib_T\otimes\Gb)\cdot\begin{bmatrix} \Acalh_1^\top \\ \vdots \\ \Acalh_k^\top \end{bmatrix} = \begin{pmatrix} \sum\limits_{j\in\J_1} p_j(\beta_1)\cdot\Acalh_j^\top \\ \vdots \\ \sum\limits_{j\in\J_n} p_j(\beta_n)\cdot\Acalh_j^\top \end{pmatrix}
\end{equation}
which are carried out locally by the servers, once they have computed their assigned tasks. We denote the encoding of the $\iota^{th}$ worker by $\Wb_\iota\in\C^{T\times N}$, i.e. $\Wb_\iota = \sum\limits_{j\in\J_\iota} p_j(\beta_\iota)\cdot\Acalh_j^\top$.

The received encoded computations by any distinct $k=n-s$ workers indexed by $\I$, constitute $\Gbt_\I\cdot(\widehat{\Ab^{-1}})^\top$. The decoding step is
{\small
\begin{align*}
  \Dbt_\I\cdot\left(\Gbt_\I\cdot(\widehat{\Ab^{-1}})^\top\right) &= \big(\Ib_T\otimes(\Gb_\I)^{-1}\big) \cdot\big(\Ib_T\otimes\Gb_\I\big)\cdot(\widehat{\Ab^{-1}})^\top\\
  &= (\Ib_T\cdot\Ib_T)\otimes\left((\Gb_\I)^{-1}\cdot\Gb_\I\right)\cdot(\widehat{\Ab^{-1}})^\top \\
  &= \Ib_T\otimes\Ib_k\cdot(\widehat{\Ab^{-1}})^\top\\
  &= (\widehat{\Ab^{-1}})^\top
\end{align*}
}
and our scheme is valid.

The above CCS therefore has a linear encoding done locally by the workers \eqref{enc_identity}, is MDS since $s=d-1$, and its decoding step reduces to computing and applying $\Gb_\I^{-1}$ (Lemma \ref{inverse_lem}). It is worth mentioning that with the above framework, any sparsest-balanced generator MDS matrix \cite{DSDY13} would suffice, as long as it satisfies the MDS theorem \cite{LX04}. By Lemma \ref{inverse_lem}, if we set $k=\Omega(\sqrt{N})$ (similar to \cite{YMAA17}), the decoding step could then be done in $\ow(N^{\omega/2})=o(N^{1.187})$, which is close to being linear in terms of $N$.

\begin{Thm}
\label{MDS_CC_thm}
Let $\Gb\in\F^{n\times k}$ be a generator matrix of any MDS code over $\F$, for which $\|\Gb^{(j)}\|_0=n-k+1$ and $\|\Gb_{(i)}\|_0=w$ for all $(i,j)\in\N_n\times\N_k$. By utilizing Algorithm \ref{inv_alg}, we can devise a linear MDS coded matrix inversion scheme; through the encoding-decoding pair $(\Gbt,\Dbt_\I)$.
\end{Thm}

Other constructions, based on cyclic MDS codes, can also be considered. These have also been leveraged to devise GC schemes \cite{RTTD17}. The corresponding encoding matrices are suitable when the network is comprised of \textit{heterogeneous} workers, as they are not sparsest-balanced.

\begin{Prop}
\label{cyclic_MDS_prop}
Any cyclic $[n,k]$ MDS code $\Cc$ over $\F\in\{\R,\C\}$ can be used to devise a coded matrix inversion encoding-decoding pair $(\Gbt,\Dbt_\I)$.
\end{Prop}

Furthermore, we can guarantee security of the communicated encodings between the workers and the master server, if we do not reveal which encoding corresponds to each worker. This is equivalent to keeping the workers' indices secret.

Observe that the rows of $\Gb$ are partitioned into $\tau=\ceil{\frac{n}{n-d}}\leqslant k$ groups of rows with the same support. By our threshold requirement that at least $k$ workers respond, the pigeonhole principle implies that at least one encoding from each of the $\tau$ groups is received. Assume the eavesdropper has knowledge of $\tau$ encoded computations; one from each group, but does not know which encoding corresponds to which group. There are a total of $\tau!$ possibilities, each of which results in a different $\Gb_\I$. This corresponds to a $\Gb$ with randomly permuted rows, $\Gb^{\text{perm}}$. Without knowledge of $\I$ and the permutation, it is then not possible to reverse the encoding $\Gbt^{\text{perm}}_\I\cdot(\widehat{\Ab^{-1}})^\top$, unless the eavesdropper exhaustively tries all $n!\tau!$ possible cases. Even in such a case, it will not know which is the correct $\widehat{\Ab^{-1}}$.

\subsection{Optimality of MDS $\brs$ Codes}
\label{opt_BRS_subs}

Under the assumption that $k=n-s$, by utilizing the $\brs_q[n,k]$ generator matrices, we achieved the minimum possible communication load from the workers to the master. From our discussion in \ref{enc_A_subs}, we cannot hope to receive an encoding of size less than $N^2/k$ when we require that $k$ workers respond with the same amount of information symbols in order to recover $\Ab^{-1}\in\R^{N\times N}$, unless we make further assumptions on the structure of $\Ab$ and $\Ab^{-1}$. Each encoding $\Wb_\iota$ consists of $NT=N^2/k$, so we have achieved the lower bound on the minimum amount of information needed to be sent to the main server by the workers. This also holds true for other generator matrices which can be used in Theorem \ref{MDS_CC_thm}, as the encodings are linear. Hence, $\Wb_\iota\in\C^{T\times N}$ for any sparsest-balanced generator MDS matrix.

We also require the workers to estimate the least possible number of columns for the given recovery threshold $k$. For our choice of parameters, the bound of \cite[Theorem 1]{TLDK17} is met with equality. That is, for all $i\in\N_n$:
$$ \|\Gb_{(i)}\|_0 = w = \frac{k}{n}\cdot d = \frac{k}{n}\cdot(n-k+1)\ , $$
which means that for homogeneous workers, we cannot get a sparser generator matrix. This, along with the requirement that $\Gb_\I$ should be invertible for all possible $\I$, are what we considered in \eqref{min_G_problem}.


\subsection{Pseudoinverse from Polynomial CMM}
\label{distr_Pseudinv_subsec}

One approach to leverage Algorithm \ref{pinv_alg} in a two-round communication scheme is to first compute $\Bb=\Ab^\top\Ab$ through a CMM scheme, then share $\Bb$ with all the workers who estimate the rows of $\widehat{\Bb^{-1}}$, and finally use another CMM to locally encode the estimated columns with blocks of $\Ab^\top$; to recover $\widehat{\Ab^{\dagger}}=\widehat{\Bb^{-1}}\cdot\Ab^\top$. Even though there are only two rounds of communication, the fact that we have a local encoding by the workers results in a higher communication load overall. An alternative approach which circumvents this issue, uses three-rounds of communication.

For this approach, we use polynomial CMM \cite{YMAA17} twice, along with our coded matrix inversion scheme. This CMM has a reduced communication load, and minimal computation is required by the workers. To have a consistent recovery threshold across our communication rounds, we partition $\Ab$ as in \eqref{parts_A} into $\kbar=\sqrt{n-s}=\sqrt{k}$ blocks. Each block is of size $N\times\Tbar$, for $\Tbar=\frac{M}{k}$. The encodings from \cite{YMAA17} of the partitions $\{\Ab_j\}_{j=1}^{\kbar}$ for carefully selected parameters $a,b\in\Z_+$ and distinct elements $\gamma_i\in\F_q$, are
\vspace{-2mm}
$$ \Abt^a_i=\sum_{j=1}^k\Ab_j\gamma_i^{(j-1)a} \quad \text{ and } \quad \Abt^b_i=\sum_{j=1}^k\Ab_j\gamma_i^{(j-1)b} $$
for each worker indexed by $i$. Thus, each encoding is comprised of $N\Tbar$ symbols. The workers compute the product of their respective encodings $(\Abt^a_i)^\top\cdot\Abt^b_i$. The decoding step corresponds to an interpolation step, which is achievable when $\kbar^2=k$ many workers respond\footnote{We select $\kbar=\sqrt{k}$ in the partitioning of $\Ab$ in \eqref{parts_A} when deploying this CMM, to attain the same recovery threshold as our inversion scheme.}, which is the optimal recovery threshold for CMM. Any fast polynomial interpolation or $\rs$ decoding algorithm can be used for this step, to recover $\Bb$.

Next, the master shares $\Bb$ with all the workers (from \ref{enc_A_subs}, this is necessary), who are requested to estimate the \textit{column-blocks} of $\widehat{\Bb^{-1}}$
\begin{equation}
\label{parts_Binv}
  \widehat{\Bb^{-1}}=\Big[\Bcalb_1 \ \cdots \ \Bcalb_k\Big] \ \ \text{ where } \Bcalb_j\in\R^{M\times\Tbar}\ \forall j\in\N_k
\end{equation}
according to Algorithm \ref{inv_alg}. We can then recover $\widehat{\Bb^{-1}}$ by our $\brs$ based scheme, once $k$ workers send their encoding.

For the final round, we encode $\widehat{\Bb^{-1}}$ as
\vspace{-2mm}
$$ \Bbt^a_i=\sum_{j=1}^k\Bcalb_j\gamma_i^{(j-1)a} $$
which are sent to the respective workers. The workers already have in their possession the encodings $\Abt^b_i$. We then carry out the polynomial CMM where each worker is requested to send back $(\Bbt^a_i)^\top\cdot\Abt^b_i$. The master server can then recover $\widehat{\Ab^{\dagger}}$.

\begin{Thm}
\label{MDS_CC_psinv_thm}
Consider $\Gb\in\F^{n\times k}$ as in Theorem \ref{MDS_CC_thm}. By using any CMM, we can devise a matrix pseudoinverse CCS by utilizing Algorithm \ref{pinv_alg}, in two-rounds of communication. By using polynomial CMM \cite{YMAA17}, we achieve this with a reduced communication load and minimal computation, in three-rounds of communication.
\end{Thm}

\section{Conclusion and Future Work}
\label{concl_sec}
\vspace{-1mm}

In this paper, we addressed the problem of computing the inverse and pseudoinverse of a matrix distributively, under the presence stragglers. Due to inherent limitations of inverting matrices, we settled for an approximation. We first gave two algorithms which respectively estimate the columns and rows of $\Ab^{-1}$ and $\Ab^{\dagger}$.

The main contribution of this work, is showing how generator matrices of sparsest-balanced MDS codes can be utilized, to devise \textit{coded matrix inversion} schemes. We worked with generator matrices of $\brs$ codes, which enables faster online decoding. A similar approach can be used to devise CMM schemes \cite{CMH21}. Furthermore, we also showed how the information can be securely transmitted between the main server and the workers, and vice versa, which is another current interest in the area of CC.

There are several interesting directions for future work. One could look into the issue of numerical stability of our $\brs$ approach, as well as if other suitable generator matrices exist. Regarding Algorithms \ref{inv_alg} and \ref{pinv_alg}, we did not establish approximation error bounds in this paper. In terms of coding theory, it would be interesting to see if it is possible to reduce the complexity of our decoding step. Specifically, could well-known $\rs$ decoding algorithms such as the Berlekamp-Welch algorithm be exploited? Another important problem is to \textit{efficiently} secure the information from the workers.

\vspace{-1mm}


\bibliographystyle{IEEEtran}
\bibliography{refs}

\begin{thebibliography}{10}
\providecommand{\url}[1]{#1}
\csname url@samestyle\endcsname
\providecommand{\newblock}{\relax}
\providecommand{\bibinfo}[2]{#2}
\providecommand{\BIBentrySTDinterwordspacing}{\spaceskip=0pt\relax}
\providecommand{\BIBentryALTinterwordstretchfactor}{4}
\providecommand{\BIBentryALTinterwordspacing}{\spaceskip=\fontdimen2\font plus
\BIBentryALTinterwordstretchfactor\fontdimen3\font minus
  \fontdimen4\font\relax}
\providecommand{\BIBforeignlanguage}[2]{{%
\expandafter\ifx\csname l@#1\endcsname\relax
\typeout{** WARNING: IEEEtran.bst: No hyphenation pattern has been}%
\typeout{** loaded for the language `#1'. Using the pattern for}%
\typeout{** the default language instead.}%
\else
\language=\csname l@#1\endcsname
\fi
#2}}
\providecommand{\BIBdecl}{\relax}
\BIBdecl

\bibitem{CPH20b}
N.~Charalambides, M.~Pilanci, and A.~O. Hero~III, ``{Straggler Robust
  Distributed Matrix Inverse Approximation},'' \emph{arXiv preprint
  arXiv:2003.02948}, 2020.

\bibitem{CPH22a}
------, ``{Secure Linear MDS Coded Matrix Inversion},'' \emph{arXiv preprint
  arxiv:2207.06271}, 2022.

\bibitem{GS59}
B.~G. Greenberg and A.~E. Sarhan, ``Matrix inversion, its interest and
  application in analysis of data,'' \emph{Journal of the American Statistical
  Association}, vol.~54, no. 288, pp. 755--766, 1959.

\bibitem{Hig02}
N.~J. Higham, \emph{Accuracy and Stability of Numerical Algorithms},
  2nd~ed.\hskip 1em plus 0.5em minus 0.4em\relax USA: Society for Industrial
  and Applied Mathematics, 2002.

\bibitem{Str69}
V.~Strassen, ``Gaussian elimination is not optimal,'' \emph{Numerische
  mathematik}, vol.~13, no.~4, pp. 354--356, 1969.

\bibitem{TB97}
L.~N. Trefethen and D.~Bau~III, \emph{Numerical linear algebra}.\hskip 1em plus
  0.5em minus 0.4em\relax Siam, 1997, vol.~50.

\bibitem{DRSL16}
T.~A. Davis, S.~Rajamanickam, and W.~M. Sid-Lakhdar, ``A survey of direct
  methods for sparse linear systems,'' \emph{Acta Numerica}, vol.~25, pp.
  383--566, 2016.

\bibitem{PV20}
R.~Peng and S.~Vempala, ``Solving sparse linear systems faster than matrix
  multiplication,'' \emph{arXiv preprint arXiv:2007.10254}, 2020.

\bibitem{AV20}
J.~Alman and V.~V. Williams, ``A refined laser method and faster matrix
  multiplication,'' \emph{arXiv preprint arXiv:2010.05846}, 2020.

\bibitem{LLPPR17}
K.~Lee, M.~Lam, R.~Pedarsani, D.~Papailiopoulos, and K.~Ramchandran, ``Speeding
  up distributed machine learning using codes,'' \emph{IEEE Transactions on
  Information Theory}, vol.~64, no.~3, pp. 1514--1529, 2017.

\bibitem{YSRKSA18}
Q.~Yu, S.~Li, N.~Raviv, S.~M.~M. Kalan, M.~Soltanolkotabi, and S.~Avestimehr,
  ``{Lagrange coded computing: Optimal design for resiliency, security and
  privacy},'' \emph{arXiv preprint arXiv:1806.00939}, 2018.

\bibitem{LA20}
S.~Li and S.~Avestimehr, ``Coded computing,'' \emph{Foundations and
  Trends{\textregistered} in Communications and Information Theory}, vol.~17,
  no.~1, 2020.

\bibitem{YMAA17}
Q.~Yu, M.~Maddah-Ali, and S.~Avestimehr, ``Polynomial codes: an optimal design
  for high-dimensional coded matrix multiplication,'' in \emph{Advances in
  Neural Information Processing Systems}, 2017, pp. 4403--4413.

\bibitem{HASH17}
W.~Halbawi, N.~Azizan, F.~Salehi, and B.~Hassibi, ``{Improving Distributed
  Gradient Descent Using {R}eed-{S}olomon Codes},'' in \emph{2018 IEEE
  International Symposium on Information Theory (ISIT)}.\hskip 1em plus 0.5em
  minus 0.4em\relax IEEE, 2018, pp. 2027--2031.

\bibitem{TLDK17}
R.~Tandon, Q.~Lei, A.~G. Dimakis, and N.~Karampatziakis, ``Gradient coding:
  Avoiding stragglers in distributed learning,'' in \emph{International
  Conference on Machine Learning}, 2017, pp. 3368--3376.

\bibitem{HLH16}
W.~Halbawi, Z.~Liu, and B.~Hassibi, ``Balanced {R}eed-{S}olomon {C}odes,'' in
  \emph{2016 IEEE International Symposium on Information Theory (ISIT)}.\hskip
  1em plus 0.5em minus 0.4em\relax IEEE, 2016, pp. 935--939.

\bibitem{HLH16b}
------, ``Balanced {R}eed-{S}olomon {C}odes for all parameters,'' in \emph{2016
  IEEE Information Theory Workshop (ITW)}.\hskip 1em plus 0.5em minus
  0.4em\relax IEEE, 2016, pp. 409--413.

\bibitem{GHSTM11}
R.~Gemulla, E.~Nijkamp, P.~J. Haas, and Y.~Sismanis, ``Large-scale matrix
  factorization with distributed stochastic gradient descent,'' in
  \emph{Proceedings of the 17th ACM SIGKDD international conference on
  Knowledge discovery and data mining}, 2011, pp. 69--77.

\bibitem{YGK17}
Y.~Yang, P.~Grover, and S.~Kar, ``Coded distributed computing for inverse
  problems,'' in \emph{Advances in Neural Information Processing Systems},
  vol.~30.\hskip 1em plus 0.5em minus 0.4em\relax Curran Associates, Inc.,
  2017, pp. 709--719.

\bibitem{RS60}
\BIBentryALTinterwordspacing
I.S.Reed and G.Solomon, ``{Polynomial Codes Over Certain Finite Fields},''
  \emph{Journal of the Society for Industrial and Applied Mathematics}, vol.~8,
  no.~2, pp. 300--304, 1960. [Online]. Available:
  \url{http://www.jstor.org/stable/2098968}
\BIBentrySTDinterwordspacing

\bibitem{FC19}
M.~Fahim and V.~R. Cadambe, ``{Numerically Stable Polynomially Coded
  Computing},'' in \emph{2019 IEEE International Symposium on Information
  Theory (ISIT)}.\hskip 1em plus 0.5em minus 0.4em\relax IEEE, 2019, pp.
  3017--3021.

\bibitem{FJHDCG17}
M.~Fahim, H.~Jeong, F.~Haddadpour, S.~Dutta, V.~Cadambe, and P.~Grover, ``On
  the optimal recovery threshold of coded matrix multiplication,'' in
  \emph{2017 55th Annual Allerton Conference on Communication, Control, and
  Computing (Allerton)}.\hskip 1em plus 0.5em minus 0.4em\relax IEEE, 2017, pp.
  1264--1270.

\bibitem{DFHJCG19}
S.~Dutta, M.~Fahim, F.~Haddadpour, H.~Jeong, V.~Cadambe, and P.~Grover, ``On
  the optimal recovery threshold of coded matrix multiplication,'' \emph{IEEE
  Transactions on Information Theory}, vol.~66, no.~1, pp. 278--301, 2019.

\bibitem{YA20}
Q.~Yu and A.~S. Avestimehr, ``{Entangled Polynomial Codes for Secure, Private,
  and Batch Distributed Matrix Multiplication: Breaking the ``Cubic''
  Barrier},'' \emph{arXiv preprint arXiv:2001.05101}, 2020.

\bibitem{KD22}
S.~Kiani and S.~C. Draper, ``{Successive Approximation Coding for Distributed
  Matrix Multiplication},'' \emph{arXiv preprint arXiv:2201.03486}, 2022.

\bibitem{CPH20a}
N.~Charalambides, M.~Pilanci, and A.~O. Hero, ``{Weighted Gradient Coding with
  Leverage Score Sampling},'' in \emph{ICASSP 2020-2020 IEEE International
  Conference on Acoustics, Speech and Signal Processing (ICASSP)}.\hskip 1em
  plus 0.5em minus 0.4em\relax IEEE, 2020, pp. 5215--5219.

\bibitem{SMA21}
M.~Soleymani, H.~Mahdavifar, and A.~S. Avestimehr, ``{Analog Lagrange Coded
  Computing},'' \emph{IEEE Journal on Selected Areas in Information Theory},
  vol.~2, no.~1, pp. 283--295, 2021.

\bibitem{BV04}
S.~P. Boyd and L.~Vandenberghe, \emph{Convex optimization}.\hskip 1em plus
  0.5em minus 0.4em\relax Cambridge university press, 2004.

\bibitem{She94}
J.~R. Shewchuk, ``{An Introduction to the Conjugate Gradient Method Without the
  Agonizing Pain},'' 1994.

\bibitem{ColNotes}
N.~Atkinson, ``Notes on the sensitivity of linear systems.''

\bibitem{CT06}
T.~M. Cover and J.~A. Thomas, \emph{Elements of Information Theory (Wiley
  Series in Telecommunications and Signal Processing)}.\hskip 1em plus 0.5em
  minus 0.4em\relax USA: Wiley-Interscience, 2006.

\bibitem{Sha79}
A.~Shamir, ``{How to Share a Secret},'' \emph{Communications of the ACM},
  vol.~22, no.~11, pp. 612--613, 1979.

\bibitem{DSDY13}
S.~H. Dau, W.~Song, Z.~Dong, and C.~Yuen, ``{Balanced Sparsest Generator
  Matrices for MDS Codes},'' in \emph{2013 IEEE International Symposium on
  Information Theory}, 2013, pp. 1889--1893.

\bibitem{Kra96}
M.~Krause, ``{A Simple Proof of the Gale-Ryser Theorem},'' \emph{The American
  Mathematical Monthly}, vol. 103, no.~4, pp. 335--337, 1996.

\bibitem{Mc01}
R.~J. McEliece, \emph{Theory of Information and Coding}, 2nd~ed.\hskip 1em plus
  0.5em minus 0.4em\relax USA: Cambridge University Press, 2001.

\bibitem{LX04}
S.~Ling and C.~Xing, \emph{{Coding Theory: A First Course}}.\hskip 1em plus
  0.5em minus 0.4em\relax Cambridge University Press, 2004.

\bibitem{RTTD17}
N.~Raviv, I.~Tamo, R.~Tandon, and A.~G. Dimakis, ``{Gradient Coding from Cyclic
  MDS Codes and Expander Graphs},'' \emph{IEEE Transactions on Information
  Theory}, vol.~66, no.~12, pp. 7475--7489, 2020.

\bibitem{CMH21}
N.~Charalambides, H.~Mahdavifar, and A.~O. Hero~III, ``Numerically stable
  binary coded computations,'' \emph{arXiv preprint arXiv:2109.10484}, 2021.

\bibitem{BP70}
{\AA}.~Bj{\"o}rck and V.~Pereyra, ``{Solution of Vandermonde Systems of
  Equations},'' \emph{Mathematics of Computation}, vol.~24, pp. 893--903, 1970.

\end{thebibliography}


\section*{Appendix 1 --- Proofs of Section \ref{Str_pr_sec}}

In this appendix, we include the missing proofs of Section \ref{Str_pr_sec}. We first recall two well-know results, which will be used.

\begin{Thm}[MDS Theorem --- \cite{LX04}]
\label{MDS_thm}
\textit{Let $\Cc$ be a linear $[n,k,d]$ code over $\F_q$, with $\Gb,\Hb$ the generator and parity-check matrices. Then, the following are equivalent}:
\begin{enumerate}
  \item \textit{$\Cc$ is a MDS code, i.e. $d=n-k+1$}
  \item \textit{every set of $n-k$ columns of $\Hb$ is linearly independent}
  \item \textit{every set of $k$ columns of $\Gb$ is linearly independent}
  \item \textit{$\Cc^{\perp}$ is a MDS code}.\\
\end{enumerate}
\end{Thm}

\begin{Thm}[BCH Bound --- \cite{HLH16},\cite{Mc01}]
\label{BCH_bd}
Let $p(\x)\in\F_q[\x]\backslash\{0\}$ with $t$ cyclically consecutive roots, i.e. $p(\alpha^{j+\iota})=0$ for all $\iota\in\N_t$. Then, at least $t+1$ coefficients of $p(\x)$ are nonzero.
\end{Thm}

\begin{proof}{[Proposition \ref{prop_sec_cryptosystem}]}
Assume for a contradiction that an adversary was able to reverse the encoding of $f(\x)$ for each block. This implies that he or she was able to reveal $\beta$ and $\Hh^{-1}$. The only way to reveal these elements, is if the adversary was able to both intercept and decipher the public-key cryptosystem used by the master, which contradicts the security of the cryptosystem.
\end{proof}

\begin{proof}{[Lemma \ref{inverse_lem}]}
The matrices $\Hb$ and $\Pb$ are of size $n\times k$ and $k\times k$ respectively. The restricted matrix $\Gb_{\I}$ is then equal to $\Hb_{\I}\Pb$, where $\Hb_{\I}\in\F_q^{k\times k}$ is now a square Vandermonde matrix, which is invertible in $\ow(k^2)$ time \cite{BP70}. Specifically
$$ \Hb_{\I} = \begin{pmatrix} 1 & \beta_{\I_1} & \beta_{\I_1}^2 & \hdots & \beta_{\I_1}^{k-1} \\ 1 & \beta_{\I_2} & \beta_{\I_2}^2 & \hdots & \beta_{\I_2}^{k-1} \\ \vdots & \vdots & \vdots & \ddots & \vdots \\ 1 & \beta_{\I_k} & \beta_{\I_k}^2 & \hdots & \beta_{\I_k}^{k-1} \end{pmatrix} \in \F_q^{k\times k} . $$
It follows that
$$ \text{det}(\Hb_{\I}) = \left(\prod\limits_{\{i<j\}\subseteq\I}(\beta_j-\beta_i)\right) $$
which is nonzero, since $\beta$ is primitive. Therefore, $\Hb_\I$ is invertible. By \cite[Lemma 1]{HLH16} and the BCH bound, we conclude that $\Pb$ is also invertible. Hence, $\Gb_\I$ is invertible for any set $\I$.

Note that the inversion of $\Pb$ can computed a priori by the master before we deploy our CCS. Therefore, computing $\Gb_\I^{-1}$ online with knowledge of $\Pb^{-1}$, requires an inversion of $\Hb_\I$ which takes $\ow(k^2)$; and then multiplying it by $\Pb^{-1}$. Thus, it requires $\ow(k^2+k^{\omega})$ operations.
\end{proof}

\begin{proof}{[Theorem \ref{MDS_CC_thm}]}
The encoding vectors applied locally by each of the $n$ workers correspond to a row of $\Gb$. The encoding by all the workers then corresponds to $\Gbt\cdot(\widehat{\Ab^{-1}})^\top$, for $\Gbt=\Ib_T\otimes\Gb$, as in \eqref{enc_identity}. Consider any set of responsive workers $\I$ of size $k$, whose encodings comprise $\Gbt_\I\cdot(\widehat{\Ab^{-1}})^\top$. By Theorem \ref{MDS_thm}, $\Gb_\I$ is invertible. Hence, the decoding step reduces to inverting $\Gb_\I$, which corresponds to $\Dbt_\I=\Ib_T\otimes(\Gb_\I)^{-1}$, and is performed online.
\end{proof}

\begin{proof}{[Lemma \ref{lem_sol_opt_prob}]}
The first two constraints are satisfied by the definition of $\Gb$, which meets the sparsest and balanced constraints with equality; for the given parameters. The last constraint is implied by 3) of Theorem \ref{MDS_thm}.

Additionally, the first two constraints of \eqref{min_G_problem} imply that $\text{nnzr}(\Gb)\geqslant\max\{nw,kd\}$, and for our parameters we have $nw=kd$. This is met with equality for the chosen $\Gb$, as
\begin{align*}
  \text{nnzr}(\Gb) &=\sum_{j\in\N_k}\text{nnzr}(\Gb^{(j)})\\
  &=\sum_{j\in\N_k}\#\big\{p_j(\beta_i)\neq0 : \beta_i\in\Bcal\big\}\\
  &=\sum_{j\in\N_k}n-\big\{i : \Mb_{ij}=0\big\}\\
  &=\sum_{j\in\N_k}n-(n-d)\\
  &=kd
\end{align*}
and the proof is complete.
\end{proof}

\begin{proof}{[Proposition \ref{cyclic_MDS_prop}]}
Consider a cyclic $[n,n-s]$ MDS code $\Cc$ over $\F\in\{\R,\C\}$. Recall that from our assumptions, we have $s=n-k$. By \cite[Lemma 8]{RTTD17}, there exists a codeword $\gb_1\in\Cc$ of support $d=s+1$, i.e. $\|\gb_1\|_0=d$. Since $\Cc$ is cyclic, it follows that the cyclic shifts of $\gb_1$ also lie in $\Cc$. Denote the $n-1$ consecutive cyclic shifts of $\gb_1$ by $\{\gb_i\}_{i=2}^n\subsetneq\Cc\subsetneq\F^{1\times n}$, which are all distinct. Define the cyclic matrix
\begin{equation}
  \Gbb \coloneqq {\begin{pmatrix} | & | & & | \\ \gb_1^\top & \gb_2^\top & \hdots & \gb_n^\top \\ | & | & & | \end{pmatrix}} \in \F^{n\times n}\ .
\end{equation}

Since $\|\gb_i\|_0=d$ and $\gb_i$ is a cyclic shift of $\gb_{i-1}$ for all $i\in\N_n$, it follows that $\|\Gbb_{(i)}\|_0=\|\Gbb_{(j)}\|_0=d$ for all $i,j\in\N_n$, i.e. $\Gbb$ is sparsest and balanced. If we erase \textit{any} $s=n-k$ columns of $\Gbb$, we get $\Gb\in\F^{n\times k}$. By erasing arbitrary columns of $\Gbb$, the resulting $\Gb$ is \textit{not} balanced\footnote{Recall that for conventional reasons we use the transpose of sparsest-balanced generator matrices, hence the \textit{balanced} condition is considered for the \textit{rows} of $\Gb$; rather than its columns.}, i.e. we have $\|\Gb_{(i)}\|_0\neq\|\Gb_{(j)}\|_0$ for some pairs $i,j\in\N_n$. Similar to the case we considered for $\brs$ generator matrices, we define the encoding matrix to be $\Gbt=\Ib_T\otimes\Gb$. The encodings are analogous to \eqref{enc_identity}.

Consider an arbitrary set of $k$ non-straggling workers $\I\subsetneq\N_n$, and the corresponding matrix $\Gb_\I\in\F^{k\times k}$. By \cite[Lemma 12, B4.]{RTTD17}, $\Gb_\I$ is invertible. The decoding matrix is then $\Dbt_\I=\Ib_T\otimes(\Gb_\I)^{-1}$, and the condition $\Dbt_\I\Gbt=\Ib_N$ is met.
\end{proof}

Next, we give a short derivation to the fact that $\tau=\ceil{\frac{n}{n-d}}\leqslant k$. In order to have a meaningful scheme, we require that $k-1\geqslant w$, otherwise every worker is assigned all computational tasks, thus everyone is requested to compute all columns of $\widehat{\Ab^{-1}}$, and a CCS is not necessary. Therefore
$$ 1-\frac{1}{k}\geqslant\frac{w}{k}=\frac{d}{n} \ \implies \ \frac{n-d}{n}\geqslant\frac{1}{k} \ \implies \ \frac{n}{n-d}\leqslant k $$
and since $k\in\Z_+$, we have $\tau\leqslant k$.


\section*{Appendix 2 --- Gradient Coding Scheme of \cite{HASH17}, and a Numerical Example}
\label{RS_GC_app}

In this appendix, we give a brief overview of the GC scheme from \cite{HASH17}, to show how it differs from our coded matrix inversion scheme. We also explicitly give their construction of a balanced mask matrix $\Mb\in\{0,1\}^{n\times k}$, which we use for the construction of the $\brs$ generator matrices. We illustrate the proposed CCS in Figure \ref{CMIS_fig}, and the encoding and decoding procedures with a simple example.

\begin{figure}[h]
  \centering
    \includegraphics[scale=.17]{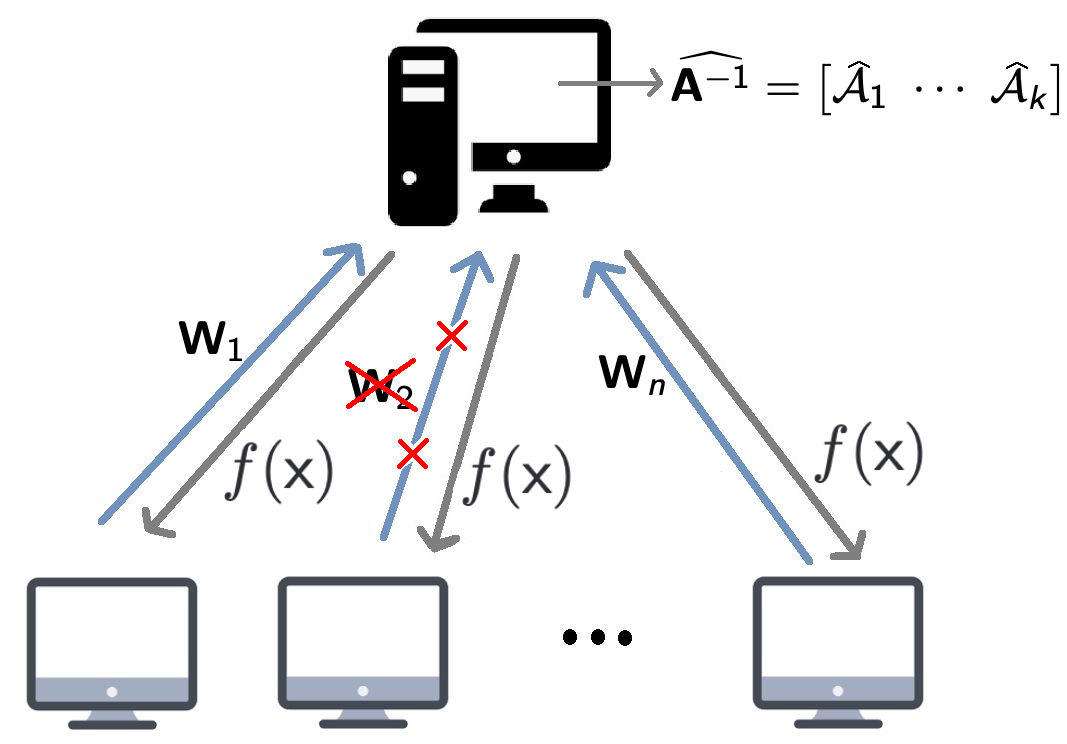}
    \caption{Algorithmic workflow of the coded matrix inversion scheme. The master shares $f(\x)$, an encoding analogous to \eqref{lagr_pol_matr}, along with $\beta,\{\eta_j^{-1}\}_{j=1}^k$. The workers then recover $\Ab$, compute their assigned tasks, and encode them according to $\Gb$. Once $k$ encodings $\Wb_\iota$ are sent back, $\widehat{\Ab^{-1}}$ can be recovered.}
  \label{CMIS_fig}
\end{figure}

\begin{algorithm}[h]
\label{RBMM}
\SetAlgoLined
\KwIn{$n,k,d\in\Z_+$ s.t. $n>d,k$ and $w=\frac{kd}{n}$}
\KwOut{row-balanced mask matrix $\Mb\in\{0,1\}^{n\times k}$}
  $\Mb \gets \bold{0}_{n\times k}$\\
  \For{$j = 0$ to $k-1$}
    {
    \For{$i = 0$ to $d-1$}
      {
        $\iota \gets (i+jd+1)\bmod n$\\
        $\Mb_{r,\iota} \gets 1$
      }
    }
 \Return $\Mb$
 \caption{MaskMatrix$(n,k,d)$ \cite{HASH17}}
\end{algorithm}

Even though this was not pointed out in \cite{HASH17}, Algorithm \ref{RBMM} does not always produce a mask matrix of the given parameters when we select $d<n/2$. This is why in our work we require $d\geqslant n/2$.

The decomposition $\Gb=\Hb\Pb$ is utilized in the GC scheme of \cite{HASH17}. Each column of $\Gb$ corresponds to a partition of the data whose partial gradient is to be computed. The polynomials are judiciously constructed in this scheme, such that the constant term of each polynomial is 1 for all polynomials, thus $\Pb_{(1)}=\vec{\bold{1}}$. By this, the decoding vector $\ab_{\I}^\top$ is the first row of $\Gb_{\I}^{-1}$, for which $\ab_{\I}^\top\Gb_{\I}=\bold{e}_1^\top$. A direct consequence of this is that $\ab_{\I}^\top\Bb_{\I}=\bold{e}_1^\top\Tb=\Tb_{(1)}=\vec{\bold{1}}$, which is the objective for constructing a GC scheme.

\subsection{Generator Matrix Example}

As an example, consider the case where $n=9$, $k=6$ and $d=6$, thus $w=\frac{kd}{n}=4$. Then, Algorithm \ref{RBMM} produces 
$$ \Mb = \begin{pmatrix} 1 & 1 & {\gray 0} & 1 & 1 & {\gray 0} \\ 1 & 1 & {\gray 0} & 1 & 1 & {\gray 0} \\ 1 & 1 & {\gray 0} & 1 & 1 & {\gray 0} \\ 1 & {\gray 0} & 1 & 1 & {\gray 0} & 1 \\ 1 & {\gray 0} & 1 & 1 & {\gray 0} & 1 \\ 1 & {\gray 0} & 1 & 1 & {\gray 0} & 1 \\ {\gray 0} & 1 & 1 & {\gray 0} & 1 & 1 \\ {\gray 0} & 1 & 1 & {\gray 0} & 1 & 1 \\ {\gray 0} & 1 & 1 & {\gray 0} & 1 & 1 \end{pmatrix} \ \in\{0,1\}^{9\times 6} \ . $$
For our CCS, this means that the $i^{th}$ worker computes the blocks indexed by $\text{supp}(\Mb_{(i)})$, e.g. $\text{supp}(\Mb_{(1)})=\{1,2,4,5\}$. We denote the indices of the respective task allocations by $\J_i=\text{supp}(\Mb_{(i)})$. The entries of the generator matrix $\Gb$ are the evaluations of the constructed polynomials \eqref{lagr_polys} at each of the  evaluation points $\Bcal=\{\beta_i\}_{i=1}^n$, i.e. $\Gb_{ij}=p_j(\alpha_i)$. This results in:
$$ \Gb = \begin{pmatrix} p_1({\beta_1}) & p_2({\beta_1}) & {\gray 0} & p_4({\beta_1}) & p_5({\beta_1}) & {\gray 0} \\ p_1({\beta_2}) & p_2({\beta_2}) & {\gray 0} & p_4({\beta_2}) & p_5({\beta_2}) & {\gray 0} \\ p_1({\beta_3}) & p_2({\beta_3}) & {\gray 0} & p_4({\beta_3}) & p_5({\beta_3}) & {\gray 0} \\  p_1({\beta_4}) & {\gray 0} & p_3({\beta_4}) & p_4({\beta_4}) & {\gray 0} & p_6({\beta_4}) \\  p_1({\beta_5}) & {\gray 0} & p_3({\beta_5}) & p_4({\beta_5}) & {\gray 0} & p_6({\beta_5}) \\  p_1({\beta_6}) & {\gray 0} & p_3({\beta_6}) & p_4({\beta_6}) & {\gray 0} & p_6({\beta_6}) \\ {\gray 0} & p_2({\beta_7}) & p_3({\beta_7}) & {\gray 0} & p_5({\beta_7}) & p_6({\beta_7}) \\ {\gray 0} & p_2({\beta_8}) & p_3({\beta_8}) & {\gray 0} & p_5({\beta_8}) & p_6({\beta_8}) \\ {\gray 0} & p_2({\beta_9}) & p_3({\beta_9}) & {\gray 0} & p_5({\beta_9}) & p_6({\beta_9}) \\ \end{pmatrix} . $$
\vspace{2mm}

\end{document}